\documentclass[12pt]{amsart}
\usepackage{txfonts}
\usepackage{amssymb}
\usepackage[mathscr]{eucal}
\usepackage{amsmath}
\usepackage{amscd}
\usepackage[dvips]{color}
\usepackage{multicol}
\usepackage[all]{xy}
\usepackage{graphicx}
\usepackage{color}
\usepackage{colordvi}
\usepackage{xspace}
\usepackage[colorlinks,final,backref=page,hyperindex,hypertex]{hyperref}

\usepackage{tikz}
\usepackage[active]{srcltx}

\begin{document}

\newcommand{\nc}{\newcommand}
\newcommand{\delete}[1]{}

\nc{\mlabel}[1]{\label{#1}}  
\nc{\mcite}[1]{\cite{#1}}  
\nc{\mref}[1]{\ref{#1}}  
\nc{\mbibitem}[1]{\bibitem{#1}} 

\delete{
\nc{\mlabel}[1]{\label{#1}  
{\hfill \hspace{1cm}{\bf{{\ }\hfill(#1)}}}}
\nc{\mcite}[1]{\cite{#1}{{\bf{{\ }(#1)}}}}  
\nc{\mref}[1]{\ref{#1}{{\bf{{\ }(#1)}}}}  
\nc{\mbibitem}[1]{\bibitem[\bf #1]{#1}} 
}

\newtheorem{theorem}{Theorem}[section]
\newtheorem{thm}[theorem]{Theorem}
\newtheorem{prop}[theorem]{Proposition}
\newtheorem{lemma}[theorem]{Lemma}
\newtheorem{coro}[theorem]{Corollary}
\newtheorem{cor}[theorem]{Corollary}
\newtheorem{prop-def}{Proposition-Definition}[section]
\newtheorem{claim}{Claim}[section]
\newtheorem{propprop}{Proposed Proposition}[section]
\newtheorem{conjecture}[theorem]{Conjecture}
\newtheorem{assumption}{Assumption}
\newtheorem{condition}[theorem]{Assumption}
\newtheorem{question}[theorem]{Question}
\theoremstyle{definition}
\newtheorem{defn}[theorem]{Definition}
\newtheorem{exam}[theorem]{Example}
\newtheorem{remark}[theorem]{Remark}
\newtheorem{ex}[theorem]{Example}
\newtheorem{coex}[theorem]{Counterexample}

\newtheorem{conv}[theorem]{Convention}

\renewcommand{\labelenumi}{{\rm(\alph{enumi})}}
\renewcommand{\theenumi}{\alph{enumi}}
\renewcommand{\labelenumii}{{\rm(\roman{enumii})}}
\renewcommand{\theenumii}{\roman{enumii}}

\nc{\Ima}{\operatorname{Im}}            
\nc{\Dom}{\operatorname{Dom}}      
\nc{\Diff}{\operatorname{Diff}}           
\nc{\End}{\operatorname{End}}        
\nc{\Id}{\operatorname{Id}}                 
\nc{\Isom}{\operatorname{Isom}}      
\nc{\Ker}{\operatorname{Ker}}           
\nc{\Lin}{\operatorname{Lin}}             
\nc{\Res}{\operatorname{Res}}         
\nc{\spec}{\operatorname{sp}}           
\nc{\supp}{\operatorname{supp}}      
\nc{\Tr}{\operatorname{Tr}}                 
\nc{\Vol}{\operatorname{Vol}}            
\nc{\sign}{\operatorname{sign}}         
\nc{\id}{\operatorname{id}}
\nc{\lin}{\operatorname{lin}}
\nc{\I}{J}
\nc{\cala}{\mathcal{A}}

\nc{\ot}{\otimes}
\nc{\bfk}{\mathbf{k}}
\nc{\wvec}[2]{{\scriptsize{\big [ \!\!
    \begin{array}{c} #1 \\ #2 \end{array} \!\! \big ]}}}

\nc{\zb}[1]{\textcolor{blue}{ #1}}
\nc{\li}[1]{\textcolor{red}{ #1}}
\nc{\sy}[1]{\textcolor{purple}{  #1}}

\newcommand{\Z}{\mathbb{Z}}
\newcommand{\ZZ}{\mathbb{Z}}
\newcommand{\Q}{\mathbb{Q}}               
\newcommand{\QQ}{\mathbb{Q}}               
\newcommand{\R}{\mathbb{R}}               
\newcommand{\RR}{\mathbb{R}}               
\newcommand{\coof}{L}           
\newcommand{\coef}{\QQ}
\newcommand{\p}{\partial}         
\nc{\cone}[1]{\langle #1\rangle}
\nc{\cl}{c}                 
\nc{\op}{o}                 
\nc{\ccone}[1]{\langle #1\rangle^\cl}
\nc{\ocone}[1]{\langle #1\rangle^\op}
\nc{\cc}{\mathfrak{C}}      
\nc{\dcc}{\mathfrak{DC}}    
\nc{\cch}{\mathfrak{Ch}}    
\nc{\dch}{\mathfrak{DCh}}   
\nc{\oc}{\mathcal{C}^o}      
\nc{\dsmc}{\dcc }  
\nc{\csup}{{^\ast}}
\nc{\bs}{\check{S}\,}
\nc{\ci}{{C,0}}
\nc{\cii}{{C,1}}
\nc{\ciii}{{C,\geq 2}}
\nc{\civ}{{IS}}
\nc{\cv}{{N}}

 \nc {\linf}{{\rm lin} (F)^\perp}
\nc{\dirlim}{\displaystyle{\lim_{\longrightarrow}}\,}
\nc{\coalg}{\mathbf{C}}
\nc{\barot}{{\otimes}}

\newcommand{\one}{\mbox{$1 \hspace{-1.0mm} {\bf l}$}}
\newcommand{\A}{\mathcal{A}}              
\newcommand{\Abb}{\mathbb{A}}          
\renewcommand{\a}{\alpha}                    
\renewcommand{\b}{\beta}                       

\newcommand{\B}{\mathcal{B}}              
\newcommand{\C}{\mathbb{C}}
 \newcommand{\calm}{{\mathcal M}}

\newcommand{\CC}{\mathcal{C}}           
\newcommand{\CR}{\mathcal{R}}           
\newcommand{\D}{\mathbb{D}}               
\newcommand{\del}{\partial}                    
\newcommand{\DD}{\mathcal{D}}           
\newcommand{\Dslash}{{D\mkern-11.5mu/\,}} 
\newcommand{\e}{\varepsilon}            
\newcommand{\F}{\mathcal{F}}                
\newcommand{\Ga}{\Gamma}                  
\newcommand{\ga}{\gamma}                   
\renewcommand{\H}{\mathcal{H}}           
\newcommand{\half}{{\mathchoice{\thalf}{\thalf}{\shalf}{\shalf}}}
\newcommand{\hideqed}{\renewcommand{\qed}{}} 
\newcommand{\K}{\mathcal{K}}             
\renewcommand{\L}{\mathcal{L}}          
\newcommand{\la}{\lambda}                   
\newcommand{\<}{\langle}
\renewcommand{\>}{\rangle}
\newcommand{\M}{\mathcal{M}}            
\newcommand{\Mop}{\star}                     
\newcommand{\N}{\mathbb{N}}             
\newcommand{\norm}[1]{\left\lVert#1\right\rVert}    
\newcommand{\norminf}[1]{\left\lVert#1\right\rVert_\infty} 
\newcommand{\om}{\omega}                 
\newcommand{\Om}{\Omega}                
\newcommand{\ol}{\\widetilde}                  
\newcommand{\OO}{\mathcal{O}}          
\newcommand{\ovc}[1]{\overset{\circ}{#1}}
\newcommand{\ox}{\otimes}                    
\newcommand{\pa}{\partial}
\newcommand{\piso}[1]{\lfloor#1\rfloor} 

\newcommand{\rad}{{\mathbf r}}
\newcommand{\sepword}[1]{\quad\mbox{#1}\quad} 
\newcommand{\set}[1]{\{\,#1\,\}}               
\newcommand{\shalf}{{\scriptstyle\frac{1}{2}}} 
\newcommand{\slim}{\mathop{\mathrm{s\mbox{-}lim}}} 
\renewcommand{\SS}{\mathcal{S}}        
\newcommand{\Sp}{{\rm Sp}}
\newcommand{\sg}{\sigma}                              
\newcommand{\T}{\mathbb{T}}                
\newcommand{\tG}{\widetilde{G}}           
\newcommand{\thalf}{\tfrac{1}{2}}            
\newcommand{\Th}{\Theta}
\renewcommand{\th}{\theta}
\newcommand{\tri}{\Delta}                        
\newcommand{\Trw}{\Tr_\omega}           
\newcommand{\UU}{\mathcal{U}}              
\newcommand{\Afr}{\mathfrak{A}}           
\newcommand{\vf}{\varphi}                       
\newcommand{\x}{\times}                          
\newcommand{\wh}{\widehat}                  
\newcommand{\wt}{\widetilde}                 
\newcommand{\ul}[1]{\underline{#1}}             
\renewcommand{\.}{\cdot}                          
\renewcommand{\:}{\colon}                       
\newcommand{\comment}[1]{\textsf{#1}}

\nc{\calc}{\mathcal{C}}
\nc{\calf}{\mathcal{F}(C\sim \cup _{i=1}^nC_i)}
\nc {\cals}{\mathcal {S}}
\nc{\calh}{\mathcal{H}}
\nc{\deff}{K}
\nc{\cali}{\mathcal{I}}
\nc{\calp}{\mathcal{P}}
\nc{\calq}{\mathcal{Q}}
\nc{\calt}{\mathcal{T}}
\nc{\vep}{\varepsilon}
\nc {\ltcone}{lattice cone\xspace}
\nc{\ltcones}{lattice cones\xspace}
\nc{\abf}{Algebraic Birkhoff Factorisation\xspace}
\nc{\abfs}{Algebraic Birkhoff Factorisations\xspace}
\nc {\lC}{(C, \Lambda _C)}
\nc {\rdim}{{\rm dim}}

\nc {\SecR}{second renormalisation}
\nc {\conefamilyc}{\underline{{C}}}
\nc {\conefamilyd}{\underline{{D}}}
\nc {\conefamilye}{\underline{{E}}}
\nc {\SSubP}{discrete open subdivision property\xspace}
\nc {\ISubP}{continuous subdivision property\xspace}
\nc {\ValP}{discrete closed subdivision property\xspace}

\title{Renormalised conical zeta values}

\author{Li Guo}
\address{Department of Mathematics and Computer Science,
         Rutgers University,
         Newark, NJ 07102, USA}
\email{liguo@rutgers.edu}

\author{Sylvie Paycha}
\address{Universit\"at Potsdam, Mathematik,
  Campus II - Golm, Haus 9
Karl-Liebknecht-Stra\ss e 24-25
D-14476 Potsdam, Germany}
\email{paycha@math.uni-potsdam.de}

\author{Bin Zhang}
\address{School of Mathematics, Yangtze Center of Mathematics,
Sichuan University, Chengdu, 610064, P. R. China}
\email{zhangbin@scu.edu.cn}

\date{\today}

\begin{abstract} Conical zeta values associated with rational convex polyhedral cones generalise
multiple zeta values. We renormalise conical zeta values  at poles by means of a generalisation of Connes and Kreimer's \abf . This paper serves as a motivation for and an application of this generalised renormalisation scheme.  The latter also yields  an Euler-Maclaurin formula on rational convex polyhedral lattice cones which relates exponential sums to exponential integrals. When restricted to Chen cones, it reduces to Connes and Kreimer's \abf for maps with values in the algebra of ordinary meromorphic functions in one variable.
\end{abstract}

\subjclass[2010]{11M32, 11H06, 52C07, 52B20, 65B15, 81T15}

\keywords{cones, coalgebras, conical zeta values, multiple zeta values, renormalisation, algebraic
Birkhoff factorisation, meromorphic functions, second renormalisation}

\maketitle
\vspace{-1.3cm}

\tableofcontents

\setcounter{section}{0}

\vspace{-1.3cm}

\allowdisplaybreaks

\section{Introduction}

Convergent conical zeta values
$$\zeta (C; \vec s):=\sum_{(n_1,\cdots,n_k)\in C\cap \Z^k}n_1^{-s_1}\cdots n_k^{-s_k},
$$
 associated with a rational convex polyhedral cone $C\subset\R^k$ and $\vec s =(s_1, \cdots, s_k)\in \ZZ ^k$, which generalise multiple zeta values, were studied in \cite{GPZ2}.    The purpose of the present paper is to study their pole structure and to evaluate them  at the poles.

 A natural idea is to apply Connes and Kreimer's \abf~\cite{CK}, see also~\cite{Ma}.  One of the  main ingredients needed for such a factorisation is  a coalgebra structure on the source space - here the space of lattice cones - of the maps to be renormalised. In~\cite {GPZ4} we showed that the space of lattice cones carries a cograded, coaugmented, connnected coalgebra structure; in the present paper, we show that this coalgebra can be enlarged to   a differential coalgebra structure (Theorem \ref{thm:HopfOnDCones}).

 Due to the geometric nature of convex cones, which is reflected in the specific coproduct built on the corresponding space of lattice cones, one cannot implement an univaluate regularisation, namely one depending on a single parameter $\e$, as Connes and Kreimer did in their \abf on Feynman graphs.  The coproduct we use  involves transverse cones built by means of an orthogonal projection, so we need a regularisation procedure which can be implemented for all cones under consideration, as well as their faces, together with the transverse cones to their faces. For a small enough family of lattice cones, such  as  the family of lattice Chen cones, their faces and the transverse lattice cones to their faces, one can  use a univaluate regularisation, in which case the regularised maps take values in Laurent series. One can then apply Connes and Kreimer's \abf  to the coalgebra of lattice cones modulo a minor adjustment due to the absence of a product on the space of such cones.  However, to deal with general convex cones and the transverse cones to their faces,   we need  (Remark \ref{rk:linearform})   a multivariate regularisation (Eq.~(\ref{eq:Somultiple})) which involves  a vector parameter $\vec \e=(\e_1,\cdots, \e_k)\in \C^k$. The regularised maps we build this way take  values in the space of multivariate meromorphic germs at zero with linear poles (Proposition-Definition \ref{defn:SI}), which we   investigated in \cite{GPZ2}.

More precisely, to renormalise conical zeta values associated   lattice cones $(C,\Lambda)$
at their poles, we implement a generalisation (Theorem
\ref{thm:Birkhoff}) of  Connes and Kreimer's \abf device \cite{CK} to the map on the coalgebra of lattice cones defined by  an exponential sum $S((C,\Lambda))$ on the lattice cones $(C,\Lambda)$. The generalisation is two fold:
\begin{itemize}
\item the  exponential sums we want to factorise  act on the colagebra of  lattice cones,  which is only equipped with a partial product, so the source space is not any longer a Hopf algebra.
\item the exponential sums  have values in the algebra of multivariate meromorphic functions, so the target space is not any longer a  Rota-Baxter algebra.
\end{itemize}

This was carried out in \cite{GPZ4}. In the present paper, we further generalise the coalgebra of cones, and consider the \abf with  additional differential structures. Indeed, in view of renormalising conical zeta values,  not only do we need to renormalise the exponential sums but also their derivatives with respect to the regularisation parameter. Hence the need for an additional differential
 structure which comes with a  decoration $\vec{s}$ leading to coloured lattice cones $(C,\Lambda,\vec s)$\footnote{Note the difference with decorated lattices cones in~\cite{GPZ2}}.

This  renormalisation procedure (Theorem \ref{th:abfd})  implemented on the exponential sums $S((C,\Lambda);\vec s)$ associated with coloured lattice cones $(C,\Lambda; \vec s)$ implies an Euler-Maclaurin formula (Eqn. \ref{eq:factSo}) on lattice cones \cite{GPZ4}  which relates exponential sums  to the corresponding  exponential integrals.
The renormalised conical zeta values $\zeta ^{\rm ren}\left((C,\Lambda),\vec s\right)$ associated with a coloured lattice cone $\left((C,\Lambda);\vec s\right)$ are derived (Eqn. (\ref{eq:renzeta}))  from the factors entering the factorisation formula of the associated exponential sum $S((C,\Lambda);\vec s)$.

On the smaller coalgebra of lattice Chen  cones,
the multivariate regularisation procedure implemented on the algebra of all convex lattice cones,  can    be reduced to a univariate regularisation procedure by specifying one direction of regularisation $\vec \e:= \vec a\, \e$ for some fixed vector $\vec a$. We show (Proposition \ref{prop:comparison}) how in the case of lattice Chen cones, specialising to  an univaluate regularisation procedure  in specifying a direction $\vec a$,  our renormalisation procedure amounts to the usual \abf on the maps given by the exponential sums on the lattice cones,  with values in Laurent series, thus independent of the choice of the direction $\vec a$. As a by-product, our geometric renormalisation procedure therefore yields renormalised multiple zeta values at negative integers obtained as renormalised conical zeta values associated with lattice Chen cones. However, these renormalised multiple zeta values do not satisfy the stuffle relations \cite{GPZ5} due to the use of the coproduct on Chen cones which involves an orthogonal complement map. Thus, the renormalised multiple zeta values we obtain here by a  geometric approach as particular instances of conical zeta values,  differ from the ones derived in \cite{MP} and \cite{GZ} by an alternative algebro-combinatorial approach. As observed in \cite{GPZ4}, the
renormalised conical values derived here by means of a multivariate \abf, can alternatively be derived directly from the derivatives of the exponential sums on  cones by means of the projection onto the holomorphic part of the meromorphic germs they give rise to.  In this respect, the
multivariate parametrisation approach-imposed here by the geometric nature of the cones- bares
over the univaluate one, the advantage that renormalisation then amounts to a projection on the
target space of multivariate meromorphic germs without the need for an \abf. So, not only is the multivariate approach necessary when dealing with the space of all cones, but it is also very useful in so far as it provides a way to circumvent the use of an \abf all together.

\section {Generalised \abf}\mlabel{sec:abd}
Let us first recall  the \abf of Connes and Kreimer's renormalisation scheme ~\mcite{CK}, which we shall then generalise in order to later renormalise conical zeta values at poles.

\begin{theorem}
Let $H$ be a commutative connected graded Hopf algebra and $(R,P)$ be a Rota-Baxter algebra of weight $-1$, $\phi: H \to R$ be an algebra homomorphism.
\begin{enumerate}
\item
There are algebra homomorphisms $\phi_-: H \to \bfk+P(R)$ and
$\phi_+: H \to \bfk+(1-P)(R)$ such that
$$\phi=\phi_-^{\ast\, (-1)}\ast \phi_+.$$
Here $\phi_-^{\ast\, (-1)}$ is the inverse of $\phi_-$ with respect to the convolution product.
\mlabel{it:decom}
\item
If $P^2=P$, then the decomposition in~(\mref{it:decom}) is unique.
\mlabel{it:uni}
\end{enumerate}
\mlabel{thm:diffBirk}
\end{theorem}

On the one hand, in \cite {GPZ4}, we generalised the \abf of Connes-Kreimer's renormalisation scheme for connected coalgebras without the need for either a Hopf algebra in the source or a Rota-Baxter algebra in the target. On the other hand, we provided the following  differential variant in~\mcite{GZ2}.
\begin {theorem}
\mlabel{th:diff}
If $(H,d)$ is in addition a differential Hopf algebra, $(R,P,\partial)$ is
a commutative differential Rota-Baxter algebra,
and $\phi$ is a differential algebra homomorphism,
then $\phi_-$ and $\phi_+$ are also differential
algebra homomorphisms.
\end {theorem}

In order to explore the structure of renormalised conical zeta values, we combine these two generalisations.

\begin {defn} A {\bf differential cograded, coaugmented, connnected coalgebra}
is a cograded, coaugmented, connnected coalgebra $\left(\coalg=\bigoplus\limits_{n\geq 0} \coalg^{(n)}, \Delta,\e, u\right)$ with linear maps $\delta_\sigma:\coalg \to \coalg$ for $\sigma$ in an index set $\Sigma$ such that
\begin{equation}
\Delta\,\delta_\sigma = (\id\ot \delta_\sigma + \delta_\sigma\ot \id)\,\Delta,\quad  \delta_\sigma(\coalg^{(n)})\subseteq \coalg^{(n+1)},\quad \delta_\sigma \,\delta_\tau=\delta_\tau\, \delta_\sigma, \quad \sigma, \tau\in \Sigma.
\mlabel{eq:dcc}
\end{equation}
The linear maps $\delta_\sigma, \sigma\in \Sigma,$ are called {\bf coderivations} on $\coalg$.   \mlabel{de:dcc}
\end {defn}
It follows from the definition that $\delta_\sigma$ stablises $\ker \e$. Recall the counit property of $\e$ for $\Delta$:
\begin{equation}
 \beta_\ell = (\e \ot \id)\Delta, \quad \beta_r =(\id\ot \e)\Delta,
 \mlabel{eq:cou}
 \end{equation}
where
$$\beta_\ell: \coalg \to \bfk \ot \coalg, x\mapsto 1 \ot x, \quad \beta_r: \coalg \to \coalg \ot \bfk, x\mapsto x\ot 1, $$ with
$$\beta_\ell^{-1}: \bfk \ot \coalg \to \coalg, a\ot x\mapsto ax, \quad \beta_r^{-1}: \coalg\ot \bfk \to \coalg, x\ot a\mapsto ax.$$

\begin{lemma}
For a differential cograded, coaugmented, connnected coalgebra $ (C,\Delta,\e,u)$ with coderivations $\delta_\sigma, \sigma\in \Sigma$, we have
$\e \delta _\sigma =0$.
\mlabel{lem:coco}
\end{lemma}
\begin{proof}
Apply $\e \ot \e$ to the two sides of the identity
$\Delta \delta_\sigma = (\id \ot \delta_\sigma +\delta_\sigma\ot \id)\Delta$. By the counit property in Eq.~(\mref{eq:cou}), on the left hand side we have
$$(\e \ot \e)\Delta \delta_\sigma = (\e \ot  \id)(\id \ot \e)\Delta\delta_\sigma = (\e \ot \id)\beta_r \delta_\sigma = (\e \delta_\sigma\ot \id)\beta_r.$$
Similarly on the right hand side we have
$$ (\e \ot \e)(\id \ot \delta_\sigma +\delta_\sigma\ot \id)\Delta
= (\e \ot \e\delta_\sigma)\Delta + (\e \delta_\sigma \ot \e)\Delta
= (1\ot \e \delta_\sigma)\beta_\ell + (\e\delta_\sigma \ot 1)\beta_r.$$
Thus we obtain
$(1\ot \e\delta_\sigma)\beta_\ell  =0$. Hence $\e\delta_\sigma =0$.
\end{proof}

As we shall argue later on, the renormalisation of conical zeta values  requires the following generalised version of this theorem \cite {GPZ4} and its differential variant, to  connected coalgebras  in the source space, which are not necessarily Hopf algebras  and algebras in the target space which are not necessarily  Rota-Baxter algebras.

\begin{thm}
Let $\coalg=\bigoplus_{n\geq 0} \coalg^{(n)}$ be a differential cograded, coaugmented, connnected coalgebra with coderivations $\delta_\sigma, \sigma\in \Sigma$\,. Let $A$ be a unitary differential algebra with derivations $\partial_\sigma, \sigma\in \Sigma.$ Let $A=A_1\oplus A_2$ be a linear decomposition  such that $1_A\in A_1$  and
$$\partial_\sigma (A_i)\subseteq A_i, \quad i=1,2, \quad \sigma\in \Sigma.$$
 Let $P$ be the projection of $A$ to $A_1$ along $A_2$.
Given  $\phi\in {\mathcal G}(\coalg,A)$ such that $\partial_\sigma \varphi = \varphi \delta_\sigma, \sigma\in \Sigma$,
define maps $\varphi_i\in {\mathcal G}(\coalg,A), i=1,2$,   by the following recursive formulae on $\ker \e$:
\begin{eqnarray}
\varphi_1(x)&=&-P\Big(\varphi(x)+\sum_{(x)} \varphi_1(x')\varphi(x'')\Big),\\
\varphi_2(x)&=&(\id_A-P)\Big(\varphi(x)+\sum_{(x)} \varphi_1(x')\varphi(x'')\Big).
\mlabel{eq:phiP}
\end{eqnarray}
\begin{enumerate}
\item
We have  $\varphi_i(\ker \e)\subseteq A_i$  $($hence  $\varphi_i: \coalg \to \bfk 1_A + A_i$$)$ and  $\delta_\sigma \varphi_i = \varphi_i \delta_\sigma, i=1, 2, \sigma\in \Sigma$. Moreover,
\begin{equation}
\mlabel{eq:abf} \varphi=\varphi_1^{\ast (-1)} \ast \varphi_2
\end{equation}
\mlabel{it:abf}
\item
$\varphi_1$ and $\varphi_2$ are are the unique maps in $ {\mathcal G}(\coalg,A)$  such that   $ \varphi_i(\ker \e)\subseteq A_i$ for $i=1, 2,$ and Eq.~(\mref{eq:abf}) holds.
\mlabel{it:uniq}
\item If moreover $A_1$ is a subalgebra of $A$ then $\varphi_1 ^{\ast (-1)}$ lies in $ {\mathcal G}(\coalg,A_1)$. \mlabel{it:phiinv}
\end{enumerate}
\mlabel{thm:abf}
\mlabel{thm:Birkhoff}
\end{thm}

\begin{remark}
When the coderivations $\delta_\sigma$ and derivations $\partial_\sigma, \sigma\in \Sigma$, are taken to be the zero maps, we obtain a generalisation of the \abf of Connes and Kreimer~\mcite{CK} which does not involve the differential structure, for maps from a connected coalgebra (which is not necessarily equipped with a product) to a decomposable unitary algebra (which does not necessarily decompose into a sum of two subalgebras). This also generalises the differential \abf in~\mcite{GZ2}.
\end{remark}

\begin{proof}
(\mref{it:abf})
The inclusion $\varphi_i(\ker \e)\subseteq A_i, i=1,2,$  follows from the definitions. Further
$$\varphi_2(x)= (\id_A -P)\Big(\varphi(x)+\sum_{(x)}\varphi_1(x')\varphi(x'')\Big)
= \varphi(x) +\varphi_1(x) +\sum_{(x)}\varphi_1(x')\varphi(x'')
= (\varphi_1\ast \varphi)(x).
$$
Since $\varphi_1(\I)= 1_A$, $\varphi_1$ is invertible for the  convolution product in $A$ as a result of ~\cite[Theorem~3.2]{GZ2} applied to   $\varphi_1$, from which  Eq.~(\mref{eq:abf}) then follows.

To verify $\partial_\sigma \varphi_i = \varphi_i \delta_\sigma, i=1,2, \sigma\in \Sigma,$ we first establish $P\partial_\sigma=\partial_\sigma P$ by verifying it on $A_1$ and $A_2$. We then implements the same inductive argument as in~\cite[Theorem~3.2]{GZ2}.

\smallskip

\noindent
(\mref{it:uniq}) Suppose there are $\psi_i\in {\mathcal G}\left(\coalg,A\right), i=1,2,$ with $  \psi_i(\ker \e)\subseteq A_i$ such that $\varphi=\psi_1^{\ast (-1)} \ast \psi_2 $.
We prove $\varphi_i(x)=\psi_i(x)$ for $i=1,2, x\in \coalg^{(k)}$ by induction on $k\geq 0$. These equations hold for $k=0$. Assume that the equations hold for $x\in\coalg^{(k)}$ where $k\geq 0$. For $x\in \coalg^{(k+1)}\subseteq \ker(\e)$, by $  \varphi_2=\varphi_1\ast \varphi$ and  $  \psi_2= \psi_1 \ast \varphi,$  we have
 $$\varphi_{2}(x)= \varphi_1(x)+ \varphi(x)+ \sum_{(x)}\varphi_1(x^\prime)\varphi(x^{\prime \prime})$$ and similarly for $\psi$, namely,
 $$\psi_{2}(x)= \psi_1(x)+ \varphi(x)+ \sum_{(x)}\psi_1(x^\prime)\varphi(x^{\prime \prime}),$$ where we have made use of
 $\varphi_1(\I)=\psi_1(\I)=\varphi(\I)=1_A$ . Hence by the induction hypothesis, we have
 $$
 \varphi_{2}(x)-\psi_{2}(x)=\varphi_{1}(x)-\psi_{1}(x)+\sum_{(x)}\big(\varphi_{1}(x^{\prime  })-\psi_{1} (x^{\prime  })\big) \varphi(x^{\prime \prime})=\varphi_{1}(x)-\psi_{1} (x)\in A_{1}\cap A_2=\{0\}.$$
 Thus
$\varphi_i(x)=\psi_i (x) \quad \text{for all } x\in \ker(\e), i=1,2.$
\smallskip

\noindent (\mref{it:phiinv}) If   $A_1$ is a subalgebra, then it follows from~\cite[Proposition~II.3.1]{Ma} applied to $A_1$ that
$\varphi_1$ is invertible in $A_1$.
\end{proof}

\delete{
\li{Give the full details as in the earlier version of \cite{GPZ4} that was omitted in the final version for limit of space?} \zb {Yes}

\begin{proof} The proof goes as in the non differential case. To verify $\partial_\sigma \varphi_i = \varphi_i \delta_\sigma, i=1,2, \sigma\in \Sigma,$ we first establish $P\partial_\sigma=\partial_\sigma P$ by verifying it on $A_1$ and $A_2$. We then implement the same inductive argument as in~\cite[Theorem~3.2]{GZ2}.   \sy{The compatibility of the coproduct with the differentiation, namely $\Delta_\sigma\circ\delta_\sigma=D_\sigma\circ \Delta$,  where $D_\sigma=\delta_\sigma\otimes 1+1\otimes \delta_\sigma$ then yields (\ref{eq:abf}).}
\end{proof}

\begin{remark}When the coderivations $\delta_\sigma$ and derivations $\delta_\sigma, \sigma\in \Sigma$, are taken to be the zero maps, this gives back the generalised  \abf derived in \cite {GPZ4}.

\end{remark}
}

\section{A differential coalgebraic structure on lattice cones}
We now apply the general setup in the last section to lattice cones.

\subsection{Lattice cones}
We begin with recalling the notion and basic properties of lattice cones. See~\cite{GPZ4} for details.
In a finite dimensional real vector space, a {\bf lattice} is a finitely generated subgroup which spans the whole space over $\R$. Such a pair, namely  a real vector space  equipped  with a lattice is called a {\bf lattice vector space}.
Let $V_{1}\subset V_2\subset \cdots $ be a family of finite dimensional real vector spaces, and let $\Lambda_{k}$ be a lattice in $V_k$ such that $\Lambda _k=\Lambda _{k+1}\cap V_k$.
The   vector space $V:=\bigcup_{k=1}^{\infty}V_{k}$ and the  corresponding lattice $\Lambda:=\bigcup_{k=1}^{\infty}\Lambda_{k}$ are equipped with their natural filtration. Such a pair $(V,\Lambda)$ is called a {\bf filtered lattice space}. Usually we work in $(\RR ^\infty, \ZZ ^\infty)$ with $V_k=\RR ^k$,  $\Lambda_k$ the standard lattice $\ZZ^k$, and $\{e_1, e_2, \cdots \}$ the standard basis.

For a filtered lattice space $(V, \Lambda)$, a point/vector in $\Lambda$ is called an {\bf lattice point/vector}, a rational multiple of an integer point/vector is called a {\bf rational lattice point/vector}.

For a subset $S$ of $V$, let $\lin(S)$ denote its $\R$-linear span. In this paper, we only consider subspaces of $V$ spanned by rational lattice vectors.

Let $V:=\cup_{k\geq1}V_k$ with $\Lambda=\cup_{k\geq 1}\Lambda_k$ be a filtered lattice space. An {\bf inner product} $Q(\cdot,\cdot)=(\cdot,\cdot)$ on $V$ is a sequence of inner products
$$ Q_k(\cdot,\cdot)=(\cdot,\cdot)_k: V_k\ot V_k \to \RR, \quad k\geq 1,$$
that is compatible with the inclusion $j_k:V_k\hookrightarrow V_{k+1}$ and whose restriction to $\Lambda\ot \QQ$ and hence $\Lambda$ takes values in $\QQ$. A filtered lattice space together with an inner product is called a {\bf filtered lattice Euclidean space}.

Let $L$ be a subspace of $V_k$.
Set
$$L^{\perp_k^Q}:=\left \{ v\in V_k\,|\, Q_k(v,u)  =0\text{ for all } u\in L\right\}.$$
The inner product $Q_k$ gives the direct sum decomposition $V_k=L\oplus L^{\perp_k^Q}$ and hence the orthogonal projection
\begin{equation} \pi_{k,L^\perp}^Q: V_k \to L^{\perp_k^Q}
\mlabel{eq:orthproj}
\end{equation}
along $L$ as well as an isomorphism
$$ \iota_{k,L}^Q: V_k/L \to L^{\perp_k^Q}.$$

Also, the induced isomorphism $Q_k^*:V_k\to V_k^*$ yields an embedding $V_k^*\hookrightarrow V_{k+1}^*$. We refer to the direct limit  $V^\circledast:=\bigcup_{k=0}^{\infty}V_{k}^*= \varinjlim V_{k}^*$   as the {\bf filtered dual space} of $V$. We will fix an inner product $Q(\cdot, \cdot)=(\cdot, \cdot)$ and drop the superscript $Q$ to simplify notations.

We collect basic definitions and facts on lattice cones that will be used in this paper, see~\cite {GPZ2} for a detailed discussion.

\begin{enumerate}
\item
By a {\bf cone} in $V_k$ we mean a {\bf closed convex (polyhedral) cone} in $V_k$, namely the convex set
\begin{equation}
\cone{v_1,\cdots,v_n}:=\RR\{v_1,\cdots,v_n\}=\RR _{\geq 0}v_1+\cdots+\RR_{\geq 0}v_n,
\mlabel{eq:cone}
\end{equation}
where $v_i\in V_k$, $i=1,\cdots, n$.
\item
A cone is called {\bf rational} if the $v_i$'s in Eq.~(\mref{eq:cone}) are in $\Lambda_k$. This is equivalent to requiring that the vectors are in $\Lambda_k\otimes \QQ$.
\item  A Chen cone is any smooth cone in $\RR ^\infty $   of the form $\langle e_1,e_1+e_2,\cdots, e_1+\cdots+e_k\rangle$ and is denoted  by $C_k^{\mathrm{Chen}}$. Note that the faces of a  Chen cone   $\langle e_1,e_1+e_2,\cdots, e_1+\cdots+e_k\rangle$ are of the form $\langle e_1+\cdots +e_{i_1}, e_1 +\cdots +e_{i_2},\cdots, e_1+\cdots +e_{i_l}\rangle$ for some indices $1\leq i_1<\cdots <i_l\leq k$,so they are not Chen cones.
\item A {\bf subdivision} of a cone $C$ is a set $\conefamilyc=\{C_1,\cdots,C_r\}$ of cones such that
\begin{enumerate}
\item[(i)] $C=\cup_{i=1}^r C_i$,
\item[(ii)] $C_1,\cdots,C_r$ have the same dimension as $C$, and
\item[(iii)] $C_1,\cdots,C_r$ intersect along their faces, i.e.,  $C_i\cap C_j$ is a face of both $C_i$ and $C_j$.
\end{enumerate}
We will use $\mathcal{F}^o(\conefamilyc )$ denote the set of  faces of $C_1,\cdots,C_r$ that are not contained in any proper face of $C$.
\item
A {\bf \ltcone}  in $V_k$ is a pair  $(C, \Lambda _C)$  with $C$  a cone in $V_k$ and $\Lambda _C$  a lattice in ${\rm lin}(C)$  generated by rational vectors.
\item A {\bf face} of a \ltcone $(C,\Lambda_C)$ is the \ltcone \ $(F, \Lambda _F)$ where $F$ is a face of $C$ and $\Lambda _F:=\Lambda _C\cap {\rm lin}(F)$.
\item
A {\bf primary generating set} of a {\ltcone} \ $\lC$ is a generating set $\{v_1,\cdots,v_n\}$ of $C$ such that
\begin{enumerate}
\item
$v_i\in \Lambda_C$, $i=1, \cdots , n$,
\item
there is no real number $r_i\in (0,1)$ such that $r_iv_i$ lies in $ \Lambda_C$, and
\item none of the generating vectors $v_i$ is a positive linear combination of the others.
\end{enumerate}
\item
A \ltcone  \ $(C, \Lambda _C)$  is called {\bf strongly convex} (resp. {\bf simplicial}) if $C$ is. A \ltcone  \ $(C, \Lambda _C)$ is called {\bf smooth} if the additive monoid $\Lambda_C\cap C$ has a monoid basis. In other words, there are linearly independent rational lattice vectors $v_1,\cdots,v_\ell$ such that
$\Lambda_C\cap C=\ZZ_{\geq 0}\{v_1,\cdots,v_\ell\}$.
\item
A {\bf subdivision} of a \ltcone  \ $(C, \Lambda _C)$  is a set of {\ltcone}s $\{(C_i, \Lambda _{C_i})\,|\, 1\leq i\leq r\}$ such that $\{C_i\,|\,1\leq i\leq r\}$ is a subdivision of $C$ and $\Lambda_{C_i}=\Lambda _C$ for all $1\leq i\leq r$.
\item
Let $F$ be a face of  a cone $C\subseteq V_k$. The {\bf transverse cone} $t(C,F)$ to $F$ is the projection
$\pi_{k,F^\perp}(C)$ of $C$ in $ \lin(F)^\perp\subseteq V_k$, where $\pi_{k,F^\perp}=\pi_{k,{\rm lin}(F)^\perp}$.
\item
Let $(F, \Lambda_F)$ be a face of the \ltcone $(C, \Lambda_C)$. The {\bf transverse \ltcone} $(t(C,F), \Lambda _{t(C,F)})$ along the face  $(F, \Lambda _F)$ is the projection of $(C, \Lambda_C)$ on $ \lin(F)^\perp\subseteq V_k$. More precisely, let $\pi_{F^\perp}:V_k\to \lin(F)^\perp$ be the projection, then
\begin{equation}
(t(C,F), \Lambda_{t(C,F)}):=(\pi_{F^\perp}(C), \pi_{F^\perp}(\Lambda_C)).
\mlabel{eq:tcdef}
\end{equation}
We also use the notation $t\left((C,\Lambda_C),(F,\Lambda_F)\right )$ to denote the transverse \ltcone.
\end{enumerate}

As in the case of ordinary cones, we have the following property.
\begin{prop} 
Any \ltcone \ can be subdivided into smooth {\ltcone}s.
\end{prop}

\begin{proof} For a given \ltcone $(D, \Lambda _C)$ in  a simplicial subdivision of a \ltcone $\lC$  with its primary generating set  $\{v_1, \cdots, v_{n}\}$,  we write $v_i=\sum\limits_{j=1}^{n} a_{ij}u_j$, $a_{ij}\in \ZZ$, $i=1, \cdots , n$, where $\{u_1, \cdots , u_{n}\}$  is  a basis of $\Lambda_C$. The absolute value of the determinant $w_D=|v_1, \cdots, v_{n}|:=|\det(a_{ij})|$ lies  in $\Z _{\ge 1}$ and is independent of the choice of a basis $\{u_1,\cdots,u_n\}$ of $\Lambda_C$. Further $w_D$ is equal to one  if and only $(D, \Lambda _C)$ is smooth.

We now prove the proposition by contradiction. Suppose $(C,\Lambda_C)$ is a  lattice cones that cannot be subdivided into smooth lattice cones.  Then for any simplicial subdivision $\conefamilyc:=\{(C_i,\Lambda_C)\}$ of $(C,\Lambda_C)$, we have
$$ w_{{\conefamilyc}}:=\max \{w_{C_i}\}>1 \quad \text{ and } \quad n_{\conefamilyc}:=\max |\{i\,|, w_{C_i}=w_{\conefamilyc}\}|\geq 1.$$

Choose a simplicial subdivision ${\conefamilyc}$ of $(C,\Lambda_C)$ with $w_{\conefamilyc}$ minimal and then among those, one with $n_{\conefamilyc}$ minimal.
We will construct a subdivision of $\lC$ that refines $\conefamilyc$. Let $D=\cone{v_1, \cdots, v_{n}}$ be a cone in ${\conefamilyc}$ with $w_D=w_{\conefamilyc}$. Since $w_D>1$, the lattice cone $(D, \Lambda_C)$ is not smooth. So $\{v_1, \cdots v_n\}$ is not a lattice basis of $\Lambda_C \cap D$. Note that the set $\{v_1,\cdots,v_{n}\}\cup \left(\left (\sum\limits_{i=1}^{n} [0,1)v_i\right) \cap \Lambda_C\right)$ spans $\Lambda_C\cap D$ as a monoid.  So there is a vector $0\neq v_D=\sum\limits_{i=1}^n c_i v_i\in \Lambda_C$ with $c_i\in [0,1)$ rational.

Reordering $v_i$, we can assume that $c_i\not=0$ for $i=1, \cdots, k,$ and $c_i=0$ for $i=k+1, \cdots , n$. We now use the vector $v_D=\sum\limits_{i=1}^{k} c_i v_i$ to subdivide the cones. Let $C_i=\cone{v_1, \cdots, v_{k}, v^{i}_{k+1}, \cdots, v^{i}_{n}}$, $i=1, \cdots , s$, be all the cones arising in the subdivision $\conefamilyc$ that contain $\cone{v_1, \cdots, v_k}$ as a face, with $C_1=D$. Then the set of cones
$$\{C_i, i>s\}\cup \{C_{ij}:=\cone{v_1,\cdots, \check{v_j}^D,\cdots,v_{k},v^{i}_{k+1}, \cdots, v^{i}_{n}}\,|\, j=1,\dots, k, \ i=1, \cdots, s\},$$
where $\check{v_j}^D$ means  $v_j$ has been replaced by $v_D$, yields a new subdivision ${\conefamilyc}'$ of $C$.

For elements in ${\conefamilyc}'$, the numbers $w_{C_i},i>s$  coincide. For $i=1,\cdots,s$ and $j=1, \cdots, k$,
$$|v_1,\cdots,\check{v_j}^D,\cdots,v_k,v^{i}_{k+1}, \cdots, v^{i}_{n} | =c_j|v_1,\cdots,v_k, v^{i}_{k+1}, \cdots, v^{i}_{n}|<|v_1,\cdots,v_k, v^{i}_{k+1}, \cdots, v^{i}_{n}|=w_{C_i}.$$
So $w_{C_{ij}}<w_{{\conefamilyc}}$. Therefore either $w_{{\conefamilyc}'}<w_{\conefamilyc}$, or $w_{{\conefamilyc}'}=w_{\conefamilyc}$ and $n_{{\conefamilyc}'}<n_{{\conefamilyc}}$. This gives the desired contradiction.
\end{proof}

\begin{prop}\zb {\cite {GPZ4}} \mlabel{pp:transversecone}
Transverse cones enjoy the following properties. Let $F$ be a face of a cone $C$.
\begin{enumerate}
\item {\bf (Transitivity) } $t(C,F)=t\left( t(C,F^\prime), t(F, F^\prime)\right)$ if $F^\prime$ is a face of $F$.
    \mlabel{it:tra}
\item {\bf (Compatibility with the partial order) } We have
$ \{H\preceq t(C,F)\} = \{t(G,F)\,|\, F\preceq G\preceq C\}.$
\mlabel{it:com1}
    \item {\bf (Compatibility with the dimension filtration) }  ${\rm dim}(C)={\rm dim} (F)+{\rm dim} \left( t(C,F) \right)$ for any face $F$ of $C$.
    \mlabel{it:com2}
\end{enumerate}
To the first two properties correspond similar properties  for lattice cones.
\begin{enumerate}\setcounter{enumi}{3}
\item {\bf (Transitivity) } $t\left((C,\Lambda_C),(F,\Lambda_F)\right)=t\left( t\left( (C,\Lambda_C),(F^\prime,\Lambda_{F^\prime})\right), t\left( (F,\Lambda_F),(F^\prime,\Lambda_{F^\prime})\right)\right)$ if $(F^\prime,\Lambda_{F^\prime})$ is a face of $(F,\Lambda_F)$.
    \mlabel{it:trad}
\item {\bf (Compatibility with the partial order) } We have
$$\left\{(H,\Lambda_H)\preceq t\left((C,\Lambda_C),(F,\Lambda_F)\right)\right\} =\left\{(t((G,\Lambda_G),(F,\Lambda_F))\,|\, (F,\Lambda_F)\preceq (G,\Lambda_G)\preceq (C,\Lambda_C)\right\}.$$
 \mlabel{it:com1d}
 \end{enumerate}
\end{prop}

\subsection{The coalgebra of {\ltcone}s}
\mlabel{subsec:coalg}
Let $\cc_k$ denote the set of {\ltcone}s in $V_k$, $k\geq 1$.
The natural inclusions $\cc_k\to \cc_{k+1}$ induced by the natural inclusions $V_k \to V_{k+1}$, $\Lambda _k\to \Lambda _{k+1}, \ k\geq 1,$
give rise to the direct  limit $\cc =\dirlim \cc_k= \cup_{k\geq 1} \cc_k$.

We equip the $\QQ$-linear space $\QQ \cc$ generated by $\cc$ with a coproduct by  means  of transverse lattice cones. The maps
\begin{equation}
\Delta: \Q\cc \longrightarrow \Q\cc \otimes \Q\cc, \quad
(C,\Lambda _C)\mapsto \sum_{F\preceq C} (t(C,F), \Lambda _{t(C,F)}) \ot (F, \Lambda _C\cap {\rm lin}(F)),
 \mlabel{eq:coproduct}
\end{equation}
\begin{equation}
\e: \Q\cc \longrightarrow \Q, \quad
(C,\Lambda _C)\longmapsto \left \{\begin{array}{ll} 1, & C=\{0\}, \\
0, & C\neq \{0\}, \end{array} \right .
\mlabel{eq:counit}
\end{equation}
and
 \begin{equation}
 u:\Q\to \Q\cc, \quad 1 \mapsto (\{0\},\{0\} ).
  \mlabel{eq:unit}
 \end{equation}
define a cograded, coaugmented, connnected coalgebra with the grading
\begin{equation}
\QQ\cc =\bigoplus_{n\geq 0} \QQ \cc ^{(n)},
\mlabel{eq:grading}
\end{equation}
where
$$\cc ^{(n)}:= \left \{ (C,\Lambda_C)\in \cc \,\big|\, \dim\,C=n\right\},\quad n\geq 0.$$

\begin {coro} For a given \ltcone \ $\lC$, the subspace
$$\bigoplus _{F\leq C} \QQ (F, \Lambda_F) \oplus \bigoplus _{F'\leq F\leq C} \QQ (t(F,F'), \Lambda _{t(F,F')})
$$
of $\QQ\cc$ is a subcoalgebra of $\QQ \cc$.
\mlabel {coro:SubCoAlg}
\end{coro}

Now we work in $(\RR ^\infty, \ZZ ^\infty)$ with $V_k=\RR ^k$,  $\Lambda_k$ the standard lattice $\ZZ^k$, and $\{e_1, e_2, \cdots \}$ the standard basis. Let $\ZZ_{\le 0} ^\infty=\dirlim \ZZ _{\le 0}^k$. For any element $\vec s=(s_i)\in \ZZ_{\le 0} ^\infty$, we set $|\vec s |:=\sum |s_i|$.

On the space  $\QQ \dcc$ freely generated by the set
$$ \dcc := \cc \times \ZZ _{\le 0}^\infty $$
of {\bf coloured lattice cones}, there is a family of linear operators
\begin{equation}
\delta_{i}: \QQ \dcc  \to \QQ \dcc  \quad ((C,\Lambda _C); \vec s)\mapsto ((C,\Lambda _C); \vec s-e_i).
\mlabel{eq:delta}
\end{equation}
By an inductive argument on $|\vec{s}|$, we obtain
\begin {lemma} For $(C,\Lambda_C)\in \cc, k\geq 1$ and $\vec{s}\in \ZZ _{\leq 0}^ {k}$, we have
$$((C,\Lambda _C); \vec s)=\delta _{1}^{-s_{1}}\cdots \delta_{k}^{-s_{k}}((C,\Lambda _C);\vec 0).$$
\mlabel{lem:diffc}
\end{lemma}

We next extend the coproduct $\Delta$ on $\QQ\cc$ to a coproduct on $\QQ \dcc $, still denoted by $\Delta$. We  proceed by induction on $n:=|\vec s |$. For $n=0$, we have $\vec s=\vec 0$ and define
$$\Delta\left((C,\Lambda _C);\vec 0\right)=\sum \left((C_{(1)},\Lambda_{C_{(1)}}),\vec 0\right)\ot \left((C_{(2)},\Lambda_{C_{(2)}}),\vec 0\right),
$$
using the coproduct $\Delta (C,\Lambda _C)=\sum (C_{(1)},\Lambda_{C_{(1)}})\ot (C_{(2)},\Lambda_{C_{(2)}})$ on $\QQ \cc$ define in Eq.~(\mref{eq:coproduct}).

Assume that the coproduct $\Delta$ has been defined for $(\lC;\vec s)$ with $|\vec s |=\ell$ for $\ell\geq 0$. Consider $(\lC,\vec s\,)\in \dcc$ with $\vec{s}\in \ZZ _{\leq 0}^ {k}$, $|\vec s |=\ell+1$. Then there is some $i$ such that $s_i\leq -1$ and  we define
\begin{equation}
\Delta(\lC;\vec s)= (\Delta \,\delta_i)(\lC;\vec s+e_i):= (D_i\, \Delta)(\lC;\vec s+e_i),
\mlabel{eq:codcalternative}
\end{equation}
where $D_i=\delta_i\otimes 1+1\otimes \delta_i$.
Explicitly, we have
\begin{equation}
\Delta (\lC;\vec s\,)=D _{1}^{-s_1}\cdots D _{k}^{-s_{k}}\Delta (\lC;\vec 0).
\mlabel{eq:codcNew}
\end{equation}
The counit $\e$ in Eq.~(\mref{eq:cou}) is trivially extended to a map on $\QQ\dcc$ for which we use the same notation
\begin{equation}
\e: \QQ\dcc  \to \QQ, \quad \e(\lC;\vec s\,)=\left\{\begin{array}{ll} 1, & (\lC;\vec s)=((\{0\},\{0\});\vec 0), \\ 0, & \text{otherwise}. \end{array} \right .
\mlabel{eq:dece}
\end{equation}
In particular, $\e $ vanishes on cones of positive dimension.
In view of the canonical embedding $\cc \to \dcc$, the unit $u$ defined in Eq.~(\mref{eq:cou}) can be seen as the map
\begin{equation}
u:\QQ \to \QQ \dsmc,\quad 1\mapsto ((\{0\},\{0\});0).
\mlabel{eq:decu}
\end{equation}

Denote
\begin{equation}
 \dcc ^{(n)}:= \left\{ (\lC;\vec s)\,\big|\, \dim\, C+|\vec s |=n\right\}, \quad n\geq 0.
\mlabel{eq:cog}
\end{equation}
Then by definition, we have $\dcc^{(0)}=\{((\{0\},\{0\});0)\}$ and $\delta_i(\dcc ^{(n)})\subseteq \dcc ^{(n+1)}, n\geq 0$.

\begin{theorem}
Let $\Delta, \e, u$ be as defined in Eqs.~$($\mref{eq:codcNew}$)$, $($\mref{eq:dece}$)$ and $($\mref{eq:decu}$)$. Equipped with the grading as in Eq.~$($\mref{eq:cog}$)$ and the derivations in Eq.~$($\mref{eq:delta}$)$, the quadruple $(\QQ\dcc,\Delta,\e,u)$ becomes a differential cograded, coaugmented, connnected coalgebra.
\mlabel{thm:HopfOnDCones}
\end{theorem}

\begin{proof}
The first  equation  in Eq.~(\mref{eq:dcc})  is just Eq.~(\mref{eq:codcalternative}). The other equations follow from the definitions.

We prove the coassociativity by induction on $|\vec s\,|$  with the initial case $|\vec{s}\,|=0$ given by the coassociativity of $\Delta $ on $\QQ \cc$, where a \ltcone \ $\lC\in \cc$ is identified with $(\lC;\vec 0)$.

Suppose the coassociativity has been proved for vectors $\vec s\in \ZZ_{\le 0}^{k}$ with $|\vec s\,|=n\geq 0$ and let $\vec s\in\ZZ _{\le 0} ^{k}$ with $|\vec s\,|=n+1$. Then there is some index $i$ with $s_i\leq -1$. By the induction hypothesis, we have $(\Delta \otimes \id) \Delta  (\lC;\vec s +\vec e_i)= (\id\otimes \Delta ) \Delta  (\lC;\vec s +\vec e_i)$. It follows that
\begin{eqnarray*} (\Delta \otimes \id)\Delta (\lC;\vec s\,)  &=& (\Delta \otimes \id)D_i\Delta  (\lC;\vec s +\vec e_i)\\
&=&(\delta _i\otimes \id\otimes \id+\id\otimes \delta _i\otimes \id+\id\otimes \id \otimes \delta _i)(\Delta \otimes \id) \Delta  (\lC;\vec s +\vec e_i)\\
&=&(\delta _i\otimes \id\otimes \id+\id\otimes \delta _i\otimes \id+\id\otimes \id \otimes \delta _i)(\id\otimes \Delta  ) \Delta  (\lC;\vec s +\vec e_i)\\
&=& (\id\otimes \Delta  )D_i\Delta  (\lC;\vec s +\vec e_i)\\
&=& (\id\otimes \Delta  ) \Delta  (\lC;\vec s\, ).
\end{eqnarray*}
This proves the coassociativity.

We  also prove the counit property $(\e \ot \id)\,\Delta = \beta_\ell$ by induction on $|\vec s|$ with the initial case $|\vec s|=0$ given by the counit property on $\QQ \cc$. Suppose that the property is proved for {\ltcone}s with $|\vec s|=\ell\geq 0$. Then for $(\lC;\vec s)\in \dcc $ with $|\vec s|=\ell+1$, there is some $1\leq i\leq k$ such that $s_i\leq -1$. Then
\begin{eqnarray*}
(\e \ot \id) \Delta(C;\vec s)
&=& (\e\ot \id)(\delta_i\ot \id + \id \ot \delta_i)\Delta(C;\vec s+e_i)\\
&=&(\e\delta_i\ot \id +\e\ot \delta_i)\Delta(C;\vec s+e_i)\\
&=&(\e\ot \delta_i)\Delta(C;\vec s+e_i)\\
&=&(\id \ot \delta_i)(\e \ot \id) \Delta(C;\vec s+e_i)\\
&=&(\id \ot \delta_i)\beta_\ell (C;\vec s+e_i)\\
&=&\beta_\ell \delta_i (C; \vec s+e_i)\\
&=&\beta_\ell (C;\vec s).
\end{eqnarray*}
This completes the induction. The proof of $(\id \ot \e)\,\Delta = \beta_r$ is similar.

From the fact that $\QQ \dcc$ is cograded with the grading in Eq.~(\mref{eq:cog}), we have
$$ \QQ \dsmc = \QQ u(1) \oplus \ker \e$$
and $\QQ\dcc ^{(0)}=\{((\{0\},\{0\});(0))\}$. Hence $\QQ\dcc$ is connected.
\end{proof}

\begin {coro} Let $\cch$ be the set of lattice Chen cones, their faces and their transverse lattice cones in $(\RR ^\infty, \ZZ ^\infty)$
and $\dch=\cch\times \ZZ _{\le 0}^\infty $ , then $\QQ \cch$ and $\QQ \dch$ are sub-coalgebras of $\QQ \dcc$.
\end{coro}
\section {Renormalisation on Chen cones}
\mlabel {sec:ChenCone}

We want to renormalise multiple zeta values, so we consider the space $\QQ \dch$.
For a \ltcone $\lC$, one way to regularise the sum
$$
\sum_{\vec n \in C^o\cap \Lambda_C}1
$$
is to introduce a linear form $\alpha$ on $V_k$ and a parameter $\e $, and then define
$$
\phi \lC:=\sum_{\vec n \in C^o\cap \Lambda_C}e^{\alpha (\vec n)\e}.
$$
Usually, we assume that $\alpha $ is rational, that is $\alpha (\vec n)\in \Q $ for $\vec n \in \Lambda _k$.

A problem arises with this regularisation, namely in order for $S\lC(\e)$ to be a Laurent series in $\e$, we need $   {\rm Ker}(\alpha)\cap C^o\cap  \Lambda_C=\{0\}$ for otherwise there are infinite many $1$'s in the summation.

\begin{remark}
\label{rk:linearform}
\begin{enumerate}
\item
For a single \ltcone, it is easy to find such a linear function $\alpha $, but problems can arise to find a linear function well suited for a family of \ltcones.
For the family $\cc$, it is impossible to find a universal $\alpha$; take any $v \in ker(\alpha)$, then $\alpha $ vanishes  on $\langle v\rangle$.
\item
For the family of cones in the the first orthant, it is also impossible to find a universal $\alpha$. This can be reduced to the two dimensional case. Any rational vector $v$ in the open upper half plane defines a cone $\langle v\rangle$ in the first quadrant or a transverse cone $\langle v\rangle=t(C,f)$ of a face $f$ of a two dimensional cone  $C$ in the first quadrant. Choosing $v$ in Ker$(\alpha)$, implies that  $\alpha $ vanishes  on $\langle v\rangle$. This extends to the closed upper half-plane since $\langle e_1\rangle$ is a cone in  the first quadrant.
\end{enumerate}
\end{remark}

However,  it is possible to find such an $\alpha $ for a small enough family, for example the family $\cch$.

\begin {prop} A linear form $\alpha =\sum a_i e_i^*$ is negative on all cones in $\cch$ if and only if $a_i<a_{i+1}<0$ for $i\in \N$.
\end{prop}
\begin {proof} In order to give the proof, we first determine the form of the transverse cones to faces of a Chen cone $C:=\langle v_1,\cdots,v_k\rangle$, where we have set $v_i:=e_1+\cdots +e_i$ for $i\geq 1$. For positive integers $p<q$, denote $[p,q]:=[p,p+1,\cdots, q]$, and $v_{[p,q]}=v_p, v_{p+1}, \cdots, v_q$. Then a face of $C$ is of the form
$$F=\langle v_{[j_0,i_1]},v_{[j_1,i_2]},\cdots,v_{[j_n,i_{n+1}]}\rangle,
\quad 0=:i_0\leq j_0\leq i_1 \bar{\leq} j_1\leq i_2\bar{\leq} j_2\leq\cdots \leq i_{n-1}\bar{\leq}j_n\leq i_{n+1}\leq j_{n+1}:=k+1.$$
Here $p\bar{\leq} q$ means $p+2\leq q$.
Then the transverse cone is generated by $\pi_{F^\perp}(v_m)$ with $i_\ell<m<j_\ell$, for $0\leq \ell \leq n+1$ with $i_\ell \bar{\leq} j_\ell$ .

First let us compute $\pi_{F^\perp}(e_m)$ for $i_\ell<m<j_\ell$, for $0\leq \ell \leq n+1$ with $i_\ell \bar{\leq} j_\ell$.
We know that
\begin {enumerate}
\item[(a1)] if $\ell=0$, $i_0 \bar{\leq} j_0$, then
$$e_m=\frac {j_0-i_0-1}{j_0 -i_0 }(e_m-e_{j_0})-\frac {1}{j_ 0-i_0 }\sum _{i_0 <t<j_0, t\not=m}(e_t-e_{j_0})+\frac {1}{j_0 -i_0 }(v_{j_0}).
$$
\item[(a2)] if $0<\ell <n+1$, $i_\ell \bar{\leq} j_\ell$, then
$$e_m=\frac {j_\ell-i_\ell-1}{j_\ell -i_\ell }(e_m-e_{j_\ell})-\frac {1}{j_\ell -i_\ell }\sum _{i_\ell <t<j_\ell, t\not=m}(e_t-e_{j_\ell})+\frac {1}{j_\ell -i_\ell }(v_{j_\ell})-\frac {1}{j_\ell -i_\ell }(v_{i_\ell}).
$$
\item[(b)] if $\ell=n+1$, $i_{n+1} \bar{\leq} j_{n+1}$, then
$$e_m=e_m.
$$
\end{enumerate}
For $0\le \ell <n+1$ and $i_\ell <t<j_\ell$, there is $(e_t-e_{j_\ell})\perp {\rm lin}(F)$. For $\ell =n+1$ and $i_{n+1} <t<j_{n+1}$, there is $e_t\perp {\rm lin}(F)$. Thus for the projection of $e_m$ we have
\begin {enumerate}
\item if $0\le \ell <n+1$, $i_\ell \bar{\leq} j_\ell$, $i_\ell<m<j_\ell$, then
 $$\pi_{F^\perp}(e_m)=\frac {j_\ell-i_\ell-1}{j_\ell -i_\ell }(e_m-e_{j_\ell})-\frac {1}{j_\ell -i_\ell }\sum _{i_\ell <t<j_\ell, t\not=m}(e_t-e_{j_\ell}).
$$
\item if $\ell =n+1$, $i_{n+1} \bar{\leq} j_{n+1}$, $i_{n+1}<m<j_{n+1}$, then
 $$\pi_{F^\perp}(e_m)=e_m.
$$
\end{enumerate}
Therefore,
\begin{enumerate}
\item if $0\le \ell <n+1$, $i_\ell \bar{\leq} j_\ell$, $i_\ell<m<j_\ell$, then
\begin{eqnarray*}
\pi_{F^\perp}(v_m)&=&\frac {j_\ell-m}{j_\ell -i_\ell }\sum _{i_\ell <t\le m}(e_t-e_{j_\ell})-\frac {m-i_\ell}{j_\ell -i_\ell }\sum _{m <t<j_\ell}(e_t-e_{j_\ell})\\
&=&\frac {j_\ell-m}{j_\ell -i_\ell }\sum _{i_\ell <t\le m}e_t-\frac {m-i_\ell}{j_\ell -i_\ell }\sum _{m <t\le j_\ell}e_t.
\end{eqnarray*}
\label{it:case1}
\item if $\ell =n+1$, $i_{n+1} \bar{\leq} j_{n+1}$, $i_{n+1}<m<j_{n+1}$, then
$\pi_{F^\perp}(v_m)=e_{i_{n+1}+1}+\cdots +e_m.$
\label{it:case2}
\end{enumerate}

We are now ready to prove the proposition, noting that $\alpha$ is negative on a transverse cone if and only if it is so on its generators $\pi_{F^\perp}(v_m), i_\ell<m<j_\ell, 0\leq \ell\leq n+1$.

Let $\alpha $ be negative on all transverse cones to faces the cone $C=\langle v_1,\cdots,v_k\rangle, k\geq 1$. Then the transverse cone for the face $\langle v_1, \cdots, \hat v_i, \cdots, v_k\rangle$ (the cone spanned by $v_1, \cdots , v_k$ except $v_i$), $i=1, \cdots , k-1$, is spanned by $\frac 12 (e_i-e_{i+1})$, by the above Case (\ref{it:case1}). Then applying $\alpha $ to this transverse cone, we have $a_i<a_{i+1}$. Now for the cone $\langle v_1, \cdots, v_{k-1}\rangle$, by Case~(\ref{it:case2}), the transverse cone is generated by $e_k$, applying $\alpha$ yields $a_k<0$. This is what we need.

Conversely, suppose that $\alpha =\sum a_ie_i^*$ satisfies $a_i<a_{i+1}<0$. Clearly, $\alpha$ is negative on $C$ and its faces. It is also negative on $\pi_{F^\perp}(v_m)$ in the Case (\ref{it:case2}). For $\pi_{F^\perp}(v_m)$ in Case~(\ref{it:case1}),
using the fact
$$\frac {j_\ell-m}{j_\ell -i_\ell }\sum _{i_\ell <t\le m}1=\frac {m-i_\ell}{j_\ell -i_\ell }\sum _{m <t\le j_\ell}1,
$$
we find
$\alpha (\pi_{F^\perp}(v_m))<0.$ Therefore $\alpha $ is negative on all transverse cones.
\end{proof}


We now fix a linear function $\alpha =\sum a_i e_i^*$ with $a_i<a_{i+1}<0$, and for $(\lC, \vec s)\in \dch$,  we set
\begin{equation}
\mlabel {eq:RegPhi}
\phi (\lC, \vec s)=\sum _{\vec n \in \Lambda _C \cap C^o}\frac {e^{\alpha (\vec n)\e}}{\vec n ^{\vec s}},
\end{equation}

Applying the same proof as for Lemma 4.4 in \cite {GPZ4}, we have

\begin {lemma}
The map $\phi(C,\Lambda_C)$ is a meromorphic function in $\e$ for any coloured lattice cone $((C,\Lambda_C),\vec{s})$ in $\dch$.
\end {lemma}

This gives rise to a linear map:
$$\phi : \Q \dch\to \C [\e ^{-1}, \e ]]
$$
to which  we can then apply Connes-Kreimer's renormalisation scheme on the coalgebra of Chen cones as in Theorem~\ref{thm:abf}, without bothering   about the product structure. So, applying the induction formula
with $(R,P)=(\C [\e ^{-1}, \e ]], -\pi_+)$, where $\pi _+$ is the projection to the holomorphic part, we have
$$\phi =\phi^{\ast (-1)}_-\ast \phi _+,
$$
where $\phi^{\ast (-1)}_-$ is the holomorphic part and $\phi _+$ is the polar part. Here $\phi _-$ takes values in $\C [[\e]]$ and $\phi _+$ takes values in $\C [\e ^{-1}]$.

Let us define renormalised multiple zeta values as
\begin{equation}\label{eq:renzeta}\zeta ^{\rm ren}(\lC ,\vec s):=\phi^{\ast (-1)}_-(\lC,\vec s)(0).
\end{equation}

We will see that the renormalised multiple zeta values do not depend on the parameters $a_i $, a fact which might seem surprising at first glance and that  will be proved in the sequel. An important consequence is that the parameters can be seen as formal parameters,  thus allowing for a regularisation in a more general situation than the one of Chen cones considered here.

\section{ Renormalised conical zeta values}

As we previously discussed, it is impossible to find a universal linear function $\alpha $  which would regularise all cones simultaneously, but it is possible to find one for the family of Chen cones; in the Chen cone case, we renormalise along a  direction $\vec a:=(a_1, a_2, \cdots) \e$. Since the parameter $\e $ can be viewed as a re-scaling of variables, this suggests  to replace the parameters  $\vec a:= (a_1, a_2, \cdots , a_k)$  by the variables
  $\vec \e  =\sum \e _ie^*_i \in V^*$, where $\e _1:= a_1\e, \e _2:=a_2\e, \cdots, \e_k:= a_k\e $ , and to  define

\begin{equation}\label{eq:Somultiple}
S^o_k (\lC;\vec s\,)(\vec \e ):= \sum _{\vec n \in \Lambda _C \cap C^o}\frac {e^{<\vec n, \vec \e >}}{\vec n ^{\vec s}}=\sum_{(n_1,\cdots,n_k)\in C^o\cap \Lambda_C} \frac{e^{n_1\e _1}\cdots e^{n_k \e _k}}{n_1^{s_1}\cdots n_k^{s_k}}= \sum_{\vec{n}\in C^o\cap \Lambda_C} \frac{e^{\langle \vec n, \vec \e \rangle}}{\vec n^{\vec{s}}}
\end{equation}
for  a simplicial \ltcone (so in particular it is strongly convex) $\lC\in \cc$  with $C\subset \RR ^k$ and where we have set
$\vec{n}^{\vec{s}}=n_1^{s_1}\cdots n_k^{s_k}$ with $\vec n:=(n_1,\cdots,n_k)\in \Lambda _C$ and $\vec s=(s_1,\cdots,s_k)\in \ZZ_{\leq 0}^k$. \\
The  sum \sy{(\ref{eq:Somultiple})}  is absolutely convergent on
$$ \check C^-\delete{=\check C^-_k}:=\Big\{\vec \e:=\sum_{i=1}^k \e_ie_i^* \,\Big|\, \langle \vec x, \vec \e \rangle <0 \text{ for all } \vec x\in C\Big\},$$
which like $C$, has  dimension $k$.
\begin{remark} With our convention that $0^s=1$ for $s$ with $\mathrm{Re}(s)\leq 0$, the function $ S^o_k(\lC;\vec s\,)(\vec \e)$ in the variables $\vec \e =\sum \e _i e_i^*$ does not depend on the choice of $k\geq 1$ such that $C\subseteq V_k$ and $\vec s \in \ZZ ^k_{\le 0}$. Thus we will suppress the subscript $k$ in the sum.
 \end{remark}

Choosing the above multivariate regularisation implies that--
in contrast to  Connes and Kreimer's renormalisation scheme-- the range space is no longer the space of Laurent series. The  new target space is a space of multivariate meromorphic germs discussed in \cite {GPZ3} which is not a Rota-Baxter algebra, thus requiring~\footnote{As observed in~\cite{GPZ4}, the renormalised conical values we derive here by means of a multivariate Algebraic Birkhoff
Factorisation, can alternatively be derived directly from the derivatives of the exponential sums on cones by means of
the projection onto the holomorphic part of the meromorphic germs they give rise to, an alternative renormalisation
method which gives rise to the same conical values.}  the  generalised version of Connes and Kreimer's renormalisation scheme corresponding to Theorem \ref{thm:Birkhoff}.


\subsection {Regularisations}

The function $S^o(\lC ,\vec s)$ is a  very specific type of meromorphic function, for it has linear poles. We briefly review the relevant definitions, and refer the reader to \cite {GPZ3} for a more detailed discussion.

\begin{defn} Let $k$ be a positive integer.
\begin{enumerate}
\item
A {\bf germ of meromorphic functions at 0} on $\C^k$ is the quotient of two holomorphic functions in a neighborhood of 0 inside $\C^k$.
\item
A germ of meromorphic functions $f(\vec \e)$ on $\C^k$ is said to have {\bf  linear poles at zero with rational coefficients} if there exist vectors $L_1, \cdots, L_n\in \Lambda_k\otimes \QQ  $ (possibly with repetitions) such that $f\,\Pi_{i=1}^n L_i$ is a holomorphic germ at zero whose Taylor expansion has rational coefficients.
\item
We will denote by $\calm_\QQ  (\C^k)$  the set of germs of meromorphic functions on $\C^k$ with linear poles at zero with rational coefficients. It is a linear subspace over $\QQ$.
\end{enumerate}
\mlabel{de:fr}
\end{defn}

Composing with the projection $\C^{k+1} \to \C^k$   dual to the inclusion $j_k:\C^k\to \C^{k+1}$ then yields the embedding
$$\calm_\coef  (\C^k)\hookrightarrow \calm_\coef  (\C^{k+1}),$$
thus giving rise to the direct limit
$$\calm_\coef  ( \C^\infty):
=\dirlim  \calm_\coef  (\C^k)
=\bigcup_{k=1}^\infty  \calm_\coef  (\C^k).
$$

\begin{prop} \cite {GPZ3}
\mlabel{prop:DecomF} There is a direct sum decomposition
$$\calm_\QQ ( \C ^\infty)=\calm _{\QQ,-}(\C ^\infty)\oplus \calm _{\QQ,+} (\C ^\infty). $$
Thus we have the projection map
\begin{equation}
\pi_+ : \calm_\QQ ( \C ^\infty) \to \calm _{\QQ,+} (\C ^\infty).
\mlabel{eq:merodec}
\end{equation}
\end{prop}


 \delete{Let $\lC\in \cc$ be a simplicial \ltcone (so in particular it is strongly convex) with $C\subset \RR ^k$. Then the set
$$ \check C^-=\check C^-_k:=\Big\{\vec \e:=\sum_{i=1}^k \e_ie_i^* \,\Big|\, \langle \vec x, \vec \e \rangle <0 \text{ for all } \vec x\in C\Big\}$$
is of dimension $k$.

Denote
$\vec{n}^{\vec{s}}=n_1^{s_1}\cdots n_k^{s_k}$ with $\vec n:=(n_1,\cdots,n_k)\in \Lambda _C$ and $\vec s=(s_1,\cdots,s_k)\in \ZZ_{\leq 0}^k$ \zb {this should be put early}.
The sums
\begin{equation}
S^o_k (\lC;\vec s\,)(\vec \e ):= \sum_{(n_1,\cdots,n_k)\in C^o\cap \Lambda_C} \frac{e^{n_1\e _1}\cdots e^{n_k \e _k}}{n_1^{s_1}\cdots n_k^{s_k}}= \sum_{\vec{n}\in C^o\cap \Lambda_C} \frac{e^{\langle \vec n, \vec \e \rangle}}{\vec n^{\vec{s}}}
 \mlabel{eq:regczv}
\end{equation}
 is absolutely convergent on $\check C^-$. By our convention of $0^s=1$ for $s$ with $\mathrm{Re}(s)\leq 0$, the function $ S^o_k(\lC;\vec s\,)(\vec \e)$ in the variables $\vec \e =\sum \e _i e_i^*$ does not depend on the choice of $k\geq 1$ such that $C\subseteq V_k$ and $\vec s \in \ZZ ^k_{\le 0}$. Thus we will suppress the subscript $k$ in the sum.}

A subdivision technique then yields the following.
\begin{prop-def}\cite {GPZ4} For any simplicial lattice cone $\lC$,
the map $S^o (\lC ;\vec s) (\vec \e )$
defines an element in $\calm_\QQ ( \C ^\infty)$.

For a general lattice cone $\lC$, the germ of functions $\sum\limits _{F\in \mathcal{F}^o(\conefamilyc )}   S^o((F,\Lambda_F);\vec s)$ does not depend on the choice of the simplicial subdivision $\conefamilyc=\{(C_i,\Lambda _{C_i})\}_{i\in [n]}$ of $\lC$. Thus we extend
(\ref{eq:Somultiple}) to any lattice cone setting
$$S^o(\lC;\vec s) :=\sum\limits _{F\in \mathcal{F}^o(\conefamilyc )}   S^o((F,\Lambda_F);\vec s),
$$
for any simplicial subdivision $\conefamilyc=\{(C_i,\Lambda _{C_i})\}_{i\in [n]}$ of $\lC$.
\mlabel{defn:SI}
\end{prop-def}

Consequently,  we have a linear map
$$S^o: \QQ\dcc \to \calm_\Q(\C^\infty), \quad (\lC;\vec s) \mapsto S^o(\lC;\vec s).$$
By definition, the following conclusion holds.

\begin{coro}Let $\lC$ be a \ltcone \ and let $\conefamilyc =\{(C_1, \Lambda _C),\cdots, (C_r, \Lambda _C) \}$ be  a  subdivision of $C$.
Then for $\vec s \in \ZZ ^k_{\le 0}$ we have
$$S^o(\lC;\vec s\,)=\sum _{F\in \mathcal{F}^o(\conefamilyc)}   S^o((F,\Lambda_C\cap {\rm lin}(F));\vec s\,)$$
in $\calm_\Q(\C^\infty)$.
\mlabel{coro:MlSub}
\end{coro}

One advantage to work with this multivariate regularisation is that the target space is stable under partial derivatives, and we thus have a linear map compatible with coderivatives..\delete{ allowing us to extend $S^o$.  to $\Q\dcc$ \zb {this is strange, we already defined $S^o$ in $\Q\dcc$ }}. Let
$$\partial _{i}=\frac {\partial }{\partial \e _i}.
$$
By an analytic continuation argument, we have the following relations between regularised conical zeta values.
\begin {prop}\mlabel{prop:partialS} For the linear map
$$S^o: \QQ\dcc \to \calm_\Q(\C^\infty)$$
and any $i\in \ZZ _{>0}$,
$$S^o \delta _i=\partial _i S^o.
$$
That means for any $(\lC, \vec s)$ in $\dcc$,
we have
$$S^o(\lC;\vec{s})( \vec \e )=  \partial^{-\vec s\,} S^o \lC( \vec \e  ),$$
where $\partial ^{-\vec s\,}=\partial_{1}^{-s_1}\cdots \partial_{ k}^{-s_k}$.
\end{prop}
\begin {proof} For a given $\vec s \in \ZZ ^k_{\le 0}$ and a simplicial {\ltcone}  $\lC\in \cc $ with $C\subset \RR ^k$,  by absolute convergence we have
$$
\partial_{i}S^o (\lC ;\vec s) (\vec \e )=S^o(\lC;\vec s-e_i)(\vec \e) =S^o(\delta _i(\lC; \vec s))(\vec \e)
$$
for $\vec \e\in\check C^-$. Therefore by analytic continuation, in $\calm_\Q(\C^\infty)$, we have
$$
\partial_{i}S^o (\lC ;\vec s) (\vec \e )=S^o(\delta _i(\lC; \vec s))(\vec \e),
$$
that is,
$$S^o \delta _i=\partial _i S^o
$$
for any simplicial lattice cone. Then by definition of $S^o$,  $S^o \delta _i=\partial _i S^o$ holds in general.
\end{proof}


\subsection {Renormalisation}
We now  equip $\RR ^\infty$  with  an inner products $Q(\cdot,\cdot)$. This allows us to construct the coalgebra $\QQ \dcc$ from transverse lattice cones introduced in Section 2, and to apply \cite[Theorem~4.2]{GPZ3} in view of the linear decomposition
$$\calm _{\QQ }(\C ^\infty)=\calm _{\QQ, + }(\C ^\infty)\oplus\calm _{\QQ, -}(\C ^\infty).$$
Since $\calm _{\Q,+}(\C ^\infty )$ is a unitary subalgebra, the \abf in Theorem~\mref{thm:abf} applies, with $\coalg=\QQ \dcc$ and  $$A= \calm_\Q (\C ^\infty),\quad A_1= \calm _{\Q,+}(\C ^\infty), \quad A_2=\calm _{\Q,-}(\C^\infty),\quad P=\pi _+:\calm_\QQ(\C ^\infty)\to \calm_{\QQ,+}(\C^\infty).$$
We consequently obtain the following theorem.
\begin {theorem} $($\text{\bf{\abf for conical zeta values}}$)$ For the linear map
$$S^o:\QQ \dcc\to \calm_\Q(\C ^\infty),
$$
there exist unique linear maps $S^o_1: \QQ \dsmc \to \calm_{\Q,+}(\C^\infty)$ and
$S^o_2: \QQ \dsmc \to \Q+\calm _{\Q,-}(\C^\infty)$, with $S^o_1(\{0\},\{0\})=1$, $S^o_2(\{0\},\{0\})=1$, such that
 \begin{equation}
 S^o= (S^o_1)^{\ast (-1)}\ast S^o_2.
\mlabel{eq:abfdC}
\end{equation}
\mlabel {th:abfd}
\end{theorem}

The same theorem applies to the sub-coalgebra $\QQ \cc$, which yields  a factorisation of $S^o:\Q\cc\to \calm_{\Q}(\C^\infty)$, giving rise to two linear maps $ S^o_1: \QQ \cc \to \calm_{\QQ,+}(\C ^\infty)$ and
$S^o_2: \QQ \cc \to \QQ+\calm _{\QQ,-}(\C ^\infty)$. We can legitimately use the same notation   as in Theorem \mref {th:abfd} since   they  correspond to the restriction of the linear maps in Theorem \mref {th:abfd} as a result of the uniqueness of the factorisation.

In \cite {GPZ4}, we identify $S^o_2$ with the exponential integral and give a formula for
$$\mu^o \lC:=(S^o_1)^{*(-1)}\lC$$
as follows.
\begin{prop} As a linear map on $\QQ \cc$, we have
$$S^o_2=I,
$$
  $$\mu^o =\pi_+ \,S^o .
  $$
\mlabel{prop:mu}
\end{prop}
Here $I$ is the exponential integral on lattice cones \cite {GPZ4} defined as follows on simplicial cones and then extended to any cone by the subdivision property. If $v_1, \cdots v_k \in \Lambda _C$ is a set of primary generators of a simplicial cone $C$, and $u_1, \cdots, u_k$  a basis of $\Lambda_C$, for $1\leq i\leq k$, let $v_i=\sum\limits_{j=1}^k a_{ji}u_j, a_{ji}\in \ZZ $. Define linear functions
$L_i:=L_{v_i}:=\sum\limits_{j=1}^k a_{ji}\langle u_j, \vec \e\rangle$ and let $w\lC $ denote the absolute value of the determinant of the matrix $[a_{ij}]$, then
\begin{equation}
I \lC (\vec \e ): =(-1)^k\frac
{w\lC}{L_1\cdots L_k}.
\mlabel{eq:ExpI}
\end{equation}


In general we also have

\begin{prop}
\mlabel{prop:DecFactorO}
For $(\lC;\vec s\,)\in \QQ \dcc $, we have
\begin {equation}
S^o_1(\lC;\vec s\,)=  \partial ^{-\vec s\,}S^o_1  \lC,\quad  \quad S^o_2(\lC;\vec s\,)=  \partial ^{-\vec s\,}S^o_2  \lC
\mlabel{eq:DecFactorO}
\end{equation}
and
\begin {equation}\mu^o =\pi_+ \,S^o .
\end{equation}
\mlabel{thm:DecFactor}
\end{prop}
\begin{proof}
By Proposition~\mref{prop:partialS} , $S^o$ are compatible with the coderivations on $\QQ\dcc$ and derivations on $\calm_\QQ(\C ^\infty)$. The conclusion then follows from Theorem~\mref{thm:abf}.
\end{proof}

For   $(\lC;\vec s\,)\in \dsmc$ the expressions $\mu^o (\lC;\vec s\,)=(S^o_1)^{*(-1)}(\lC;\vec s\,)$ in the \abf of $S^o$ is a germ of holomorphic functions which we  can  therefore  evaluate   at $ 0$.

\begin {defn} The value $$\zeta^o (\lC;\vec s\,):=(S^o_1)^{*(-1)}(\lC;\vec s\,)(0)$$
is called the {\bf renormalised open conical zeta value} of $(\lC;\vec s\,)$.
\end {defn}

In particular, this definition applies to cones in $\cch$ and $\dch$.

\begin{cor} The germs of functions $(S^o_1)^{*(-1)}\lC$ are generating functions of renormalised open  conical zeta values at nonpositive integers.  More precisely, for a   \ltcone \  $\lC\in \cc$,  we have
\begin{equation}\mlabel{eq:Taylorkdimc} (S^o_1)^{*(-1)} \lC(\vec \e)=\sum_{\vec r\in \Z_{\geq 0}^k}^\infty \zeta^o(\lC;-\vec r\,) \frac{\vec \e^{\,\vec r}}{\vec r!}.\end{equation}
\mlabel{pp:diffrenzeta}
 \end{cor}

\begin {proof} By Eq.~(\mref{eq:DecFactorO}), we have
$$\partial_{\vec \e}^{\vec r\,}(S^o_1)^{*(-1)}\lC(0)=(S^o_1)^{*(-1)}(\lC;-\vec r\,)(0)=\zeta^o(\lC;-\vec r\,),$$
as needed.
\end{proof}

\section {Comparison of the two renormalisation schemes}
So far, we have two approaches to renormalise sums on Chen cones, which can be related by means of a restriction  $\vec \e =\vec a\, \e$ along a direction $\vec a$: the first one by which the \abf procedure is implemented after restricting, the second one by which the \abf procedure is implemented before  restricting.

 Under the restriction along a direction $\vec a$, the splittings of the target space in the two approaches differ as it can be seen on the following counterexample which shows that evaluation  ${\mathcal E}_{\vec a}$ along a given direction $\vec a\, \e$ does not commute with the projection $\pi_+$:
$$\pi_+\circ {\mathcal E}_{\vec a}\neq {\mathcal E}_{\vec a}\circ \pi_+,$$
where the projection $\pi_+$ on the left hand side is the one on  $\calm_\QQ(\C ^\infty)$ and the one on the right hand side is on $\calm_\QQ(\C )$.

\begin{coex}  Let   $f(\e_1,\e_2):=\frac{\e_1}{ \e_2}$, then
$$\pi_+\circ {\mathcal E}_{\vec a}(f)= \frac{a_1}{a_2}\neq 0= {\mathcal E}_{\vec a}\circ \pi_+(f).$$
\end{coex}

 But surprisingly, these two renormalisation procedures give the same renormalised values for Chen cones.

\begin {prop}\label{prop:comparison} For Chen cones, the factorisations obtained by
\begin{itemize}
\item first implementing the \abf on the exponential sum $S^o$ and then restricting along a direction $\vec a \e $, and
\item first restricting the exponential sum $S^o$ along a direction $\vec a \e$  and then implementing the \abf
\end{itemize}
coincide.
\end{prop}

\begin {proof} We first investigate the first renormalisation procedure.  Since  the \abf applied to  the exponential sum $S^o$ on cones  boils down to the Euler-Maclaurin formula on cones \cite {GPZ4}, we have that on $\QQ \cc$
\begin{equation}\label{eq:factSo} S^o=\mu ^o\ast I,
\end{equation}
where $\ast$ is the convolution associated with the coproduct on lattice cones. For any lattice cone $\lC$ ,     $\mu^o \lC $ is holomorphic and $I\lC$ is  a sum of simple fractions. By Proposition \mref{prop:DecFactorO}, differentiating  yields for any lattice cone $\lC$ and any $\vec s$, a holomorphic function $\mu^o(\lC; \vec s)$ and a sum $I(\lC; \vec s)$ of simplicial fractions.
Now, restricting  along the direction $\vec \e =\vec a\, \e$ yields for any lattice cone $\lC$ and $\vec s$,  a map $\mu^0 (\lC;\vec s)\vert_{\vec \e =\vec a\, \e} $ in $\QQ [[\e]]$. Furthermore,    the restriction $I (\lC;\vec s)\vert_{\vec \e =\vec a\, \e}$  lies in $ \QQ [\e ^{-1}]\e ^{-1}$ if $(\lC;\vec s)\not =((\{0\},\{0\}),\vec 0)$ as a sum of  restricted simplicial fractions. So if we let
$$\tilde {\mu }  (\lC;\vec s)(\e )=\mu ^o(\lC;\vec s)(\vec \e )\vert_{\vec \e =\vec a\, \e} ,
$$
and
$$\tilde {I}  (\lC;\vec s)(\e )=I(\lC;\vec s)(\vec \e )\vert_{\vec \e =\vec a\, \e} ,
$$
with  $\phi  (\lC;\vec s)(\e )=S^o(\lC;\vec s)(\vec \e )\vert_{\vec \e =\vec a\, \e}$ as in (\mref {eq:RegPhi}), we have
$$\phi =\tilde \mu \ast \tilde I,
$$
where  $\tilde \mu(\lC; \vec s)\in \QQ  [[\e]]$ and $\tilde I(\lC; \vec s)\in \QQ +\QQ  [\e ^{-1}]\e ^{-1}$.

The alternative renormalisation procedure is to implement  \abf   on the restricted map $\phi$, which yields  a factorisation
$$\phi =\phi^{\ast (-1)}_-\ast \phi _+,
$$
with $\phi^{\ast (-1)}_- (\lC;\vec s)\in \C [[\e]]$, and $ \phi _+ (\lC;\vec s)\in \C [\e ^{-1}]$.

Thus both factorisations are for linear maps between the same spaces. Now   the standard argument of the uniqueness of the \abf then shows that the two factorisations coincide.
\end {proof}

\begin {coro} The renormalised multiple zeta values do not depend on the parameters $a_1, a_2, \cdots $.
\end{coro}

Let us illustrate the two approaches on a simple example.
To simplify notations, for $k$ linear forms $L_1, \cdots, L_k$, we set

\begin{equation}
 \label{eq:notationL}[L_1, \cdots , L_k]:=\frac {e^{L_1}}{1-e^{L_1}}\frac {e^{L_1+L_2}}{1-e^{L_1+L_2}}\cdots \frac {e^{L_1+L_2+\cdots +L_k}}{1-e^{L_1+L_2+\cdots +L_k}}.
\end{equation}
and
\begin{equation}  \frac {e^{\e}}{1-e^{\e}} =-\frac{1}{\e}+h(\e).
\end{equation}
\begin{ex}
For  $k=2$ and the Chen cone $<e_1, e_1+e_2>$, we have
$$S^o(<e_1, e_1+e_2>, \Lambda_2)=[\e_1, \e_2],
$$
\begin{eqnarray}\label{eq:k2}\pi_+\left([\e_1, \e_2] \right)    &=&\pi_+\left(\Big(-\frac 1{\e_1 }+h(\e _1 )\Big)\Big(-\frac 1{\e_1+\e_2}+h(\e _1+\e _2)\Big)\right)\nonumber\\
&=& \pi_+\left(-\frac {h(\e _1+\e_2 )}{\e_1 }-\frac {h(\e _1 )}{\e_1+\e _2 }+h(\e _1 )h(\e _1+\e_2)\right)\nonumber\\
&=& -\frac {h(\e _1+\e_2 )-h(\e _2)}{\e_1}-\frac {h(\e _1 )-h\left(\frac {\e_1-\e_2}2\right)}{\e_1+\e _2 }+h(\e _1 )h(\e _1+\e_2 ).\notag
\end{eqnarray}
So
$$\pi_+\left([\e_1, \e_2] \right) \vert_{ (a_1\e, a_2\e)}=
 -\frac {h((a_1+a_2)\e   )-h(a_2\e )}{a_1\e}-\frac {h(a_1\e  )-h\left(\frac {(a_1-a_2)\e}2\right)}{(a_1+a_2)\e }+h(a_1\e  )h((a_1+a_2)\e).$$
 Evaluating at $\e=0$ yields
 $$\zeta(0,0)= -\frac{(a_1+a_2)-a_2}{a_1}h^\prime(0)-\frac{\frac{a_1+a_2}{2}}{a_1+a_2}\, h^\prime(0)+h(0)^2=   -\frac{3}{2}\, h^\prime(0)+h(0)^2= \frac 38.$$

 On the other hand, to use formula~(\ref {eq:renzeta}) to find $\phi_- ^{\ast (-1)}$ needs more involved computations.
We easily get
 $$\phi_-^{\ast (-1)}(<e_1>, \ZZ e_1)=h (a_1\e),
 $$
 and
 $$\phi_-^{\ast (-1)}(<e_1+e_2>, \ZZ (e_1+e_2))=h((a_1+a_2)\e).
 $$

 The reduced coproduct applied to the two dimension Chen cone reads
 $$\Delta^\prime (\langle e_1,e_1+e_2\rangle, \Lambda _2)= (\langle e_2\rangle, \ZZ e_2) \otimes (\langle e_1\rangle, \ZZ e_1) + (\langle e_1-e_2\rangle, \ZZ \frac {e_1-e_2}2) \otimes (\langle e_1+e_2\rangle, \ZZ (e_1+e_2)).$$
Thus
 \begin{eqnarray*}
 &&\phi_-(<e_1, e_1+e_2>, \Lambda_2)\\
 &=&-P\Big(\Big(-\frac 1{a_1\e }+h(a _1\e )\Big)\Big(-\frac 1{(a_1+a_2)\e }+h((a _1+a _2)\e)\Big)\\
 &&+\Big(-h(a_2\e)\Big)\,\Big(-\frac 1{a_1\e }+h(a _1\e )\Big) + \Big(-h((a_1-a_2)\e /2 )\Big)\,\Big(-\frac 1{(a_1+a_2)\e }+h((a _1+a _2)\e)\Big)\Big)\\
&=&\frac {h((a_1+a_2) \e)-h(a_2 \e)}{a_1\e }+\frac {h(a _1\e )-h((a_1-a_2)\e /2 )}{(a_1+a_2)\e}-h(a_1\e )h((a_1+a_2)\e)\\
&&+h(a_2\e)h (a_1\e)  +h((a_1-a_2)\e /2 )h((a _1+a _2)\e).
\end{eqnarray*}
Now by the equation
\begin{eqnarray*}
&&\phi_-(<e_1, e_1+e_2>, \Lambda_2)+\phi_-^{\ast (-1)}(<e_1, e_1+e_2>, \Lambda_2)\\
&+&\phi_-(<e_2>, \ZZ e_2)\phi_-^{\ast (-1)}(<e_1>, \ZZ e_1)+\phi_-(\langle e_1-e_2\rangle, \ZZ \frac {e_1-e_2}2)\phi_-^{\ast (-1)}(<e_1+e_2>, \ZZ (e_1+e_2))\\
&=&0,
\end{eqnarray*}
we have
$$\phi_-^{\ast (-1)}(<e_1, e_1+e_2>, \Lambda_2)=
 -\frac {h((a_1+a_2)\e   )-h(a_2\e )}{a_1\e}-\frac {h(a_1\e  )-h\left(\frac {(a_1-a_2)\e}2\right)}{(a_1+a_2)\e }+h(a_1\e  )h((a_1+a_2)\e).$$
This agrees with $\pi_+\left([\e_1, \e_2] \right) \vert_{ (a_1\e, a_2\e)}$.
\end{ex}
\smallskip

\noindent
{\bf Acknowledgements}:
This work is supported by the National Natural Science Foundation of China (Grant No. 11071176, 11221101 and 11371178) and the National Science Foundation of US (Grant No. DMS~1001855). The authors thank Kavli Institute for Theoretical Physics China (KITPC) and Morningside Center of Mathematics (MCM) in Beijing where part of the work was carried out.
The second author thanks Sichuan University, Lanzhou University and Capital Normal University for their kind hospitality.


\begin{thebibliography}{abcdsfgh}





\bibitem{CK} A. Connes and D. Kreimer, Hopf algebras,
    Renormalisation and Noncommutative Geometry, {\em Comm. Math. Phys.} {\bf  199} (1988) 203-242.








\bibitem{GPZ2} L. Guo, S. Paycha and B. Zhang, Conical zeta values and their double subdivision relations, {\em Adv. Math.} {\bf 252} (2014) 343-381.

\bibitem{GPZ3} L. Guo, S. Paycha and B. Zhang, Residue of meromorphic functions with linear poles, preprint

\bibitem{GPZ4} L. Guo, S. Paycha and B. Zhang, Algebraic Birkhoff Factorisation and the Euler-Maclaurin formula on cones, arXiv:1306.3420.


\bibitem{GPZ5} L. Guo, S. Paycha and B. Zhang, Counting  an infinite number of points: a testing ground for renormalisation methods, arXiv:1501.00429.

\bibitem{GZ} L. Guo and B. Zhang, Renormalisation of multiple zeta
    values {\em J. Algebra} {\bf 319} (2008) 3770-3809.

\bibitem{GZ2} L. Guo and B. Zhang, Differential Birkhoff decomposition and renormalisation of multiple zeta values, {\em J. Number Theory}\, {\bf 128} (2008), 2318-2339.







\bibitem{Ma} D. Manchon, Hopf algebras, from basics to applications to renormalisation, {\em Comptes-rendus des Rencontres math\'ematiques de Glanon} 2001 (2003);
 {\it Hopf algebras in renormalisation}, Handbook of algebra, Vol. 5 (M. Hazewinkel ed.) (2008).

\bibitem{MP} D. Manchon and S. Paycha, Nested sums of symbols and renormalised multiple zeta values, Int. Math. Res. Papers 2010 issue 24, 4628-4697 (2010).



\end{thebibliography}
\end{document}